\documentclass[preliminary,creativecommons]{eptcs}
\providecommand{\event}{WORDS 2011}
\usepackage{breakurl}
\usepackage[english]{babel}
\usepackage{hyperref,tikz,amsmath,amsthm,amsfonts}

\newtheorem{theorem}{Theorem}
\newtheorem{proposition}{Proposition}
\newtheorem{lemma}{Lemma}
\theoremstyle{remark}
\newtheorem{example}{Example}

\title{On the Delone property of ($-\beta$)-integers}

\author{Wolfgang Steiner
\institute{LIAFA, CNRS, Universit\'e Paris Diderot -- Paris 7, \\ Case 7014, 75205 Paris Cedex 13, France}
\email{steiner@liafa.jussieu.fr}}

\def\titlerunning{On the Delone property of ($-\beta$)-integers}
\def\authorrunning{W. Steiner}

\begin{document}
\maketitle

\begin{abstract}
The $(-\beta)$-integers are natural generalisations of the $\beta$-integers, and thus of the integers, for negative real bases. 
They can be described by infinite words which are fixed points of anti-morphisms.
We show that they are not necessarily uniformly discrete and relatively dense in the real numbers.
\end{abstract}

\section{Introduction}
We study the set of \emph{$(-\beta)$-integers} for a real number $\beta > 1$.
This set is defined by
\[
\mathbb{Z}_{-\beta} = \bigcup_{n\ge0} (-\beta)^n\, T_{-\beta}^{-n}(0)\,,
\]
where $T_{-\beta}$ is the \emph{$(-\beta)$-transformation}, defined by Ito and Sadahiro \cite{Ito-Sadahiro09} as
\[
T_{-\beta}:\ \big[\tfrac{-\beta}{\beta+1}, \tfrac{1}{\beta+1}\big), \quad x \mapsto -\beta x - \big\lfloor \tfrac{\beta}{\beta+1} - \beta x \big\rfloor\,.
\]
Equivalently, a $(-\beta)$-integer is a real number of the form
\[
\sum_{k=0}^{n-1} d_k\, (-\beta)^k\,, \quad \mbox{with} \quad \frac{-\beta}{\beta+1} \le \sum_{k=0}^{m-1} d_k\, (-\beta)^{k-m} < \frac{1}{\beta+1} \quad \mbox{for all}\ 1 \le m \le n,
\]
where $d_0, d_1, \ldots, d_{k-1}$ are integers.
Examples of $(-\beta)$-transformations are depicted in Figure~\ref{f:Tmbeta}.

Recall that the set of \emph{$\beta$-integers} is defined by
\[
\mathbb{Z}_\beta = \mathbb{Z}_\beta^+ \cup (-\mathbb{Z}_\beta^+) \quad \mbox{with} \quad \mathbb{Z}_\beta^+ = \bigcup_{n\ge0} \beta^n\, T_\beta^{-n}(0)\,,
\]
where $T_\beta$ is the \emph{$\beta$-transformation},
\[
T_\beta:\,[0,1) \to [0,1)\,,\quad x \mapsto \beta x - \lfloor \beta x \rfloor\,.
\]
These sets were introduced in the domain of quasicrystallography, see e.g.~\cite{Burdik-Frougny-Gazeau-Krejcar98}.

It is not difficult to see that $\mathbb{Z}_{-\beta} = \mathbb{Z}$ when $\beta \in \mathbb{Z}$, and that $\mathbb{Z}_{-\beta} = \{0\}$ when $\beta < \frac{1+\sqrt5}{2}$.
For $\beta \ge \frac{1+\sqrt5}{2}$, Ambro\v{z} et al.~\cite{Ambroz-Dombek-Masakova-Pelantova} showed that $\mathbb{Z}_{-\beta}$ can be described by the fixed point of an anti-morphism on a possibly infinite alphabet.
They also calculated explicitely the set of distances between consecutive $(-\beta)$-integers when $T_{-\beta}^n\big(\frac{-\beta}{\beta+1}\big) \le 0$ and $T_{-\beta}^{2n-1}\big(\frac{-\beta}{\beta+1}\big) \ge \frac{1-\lfloor\beta\rfloor}{\beta}$ for all $n \ge 1$.
It seems to be difficult to extend their methods to the general case.
For the case when $\beta$ is an \emph{Yrrap number}, i.e., when $\big\{T_{-\beta}^n\big(\frac{-\beta}{\beta+1}\big)  \mid n \ge 0\big\}$ is a finite set, a~different approach can be found in~\cite{Steiner}.
The approach in Section~\ref{sec:fixed-point-an} resembles that in~\cite{Steiner}, but it is simpler and works for general~$\beta$.
In Section~\ref{sec:delone-property}, we discuss the Delone property for sets~$\mathbb{Z}_{-\beta}$. 

\section{Fixed point of an anti-morphism} \label{sec:fixed-point-an}
By Lemma~\ref{l:Ziota}, we can consider the set of $(-\beta)$-integers, $\beta > 1$, as a special instance of the preimage of a point in $\big[\tfrac{-\beta}{\beta+1}, \tfrac{1}{\beta+1}\big)$ of the map
\[
\iota_\beta:\, \mathbb{R} \to \big[\tfrac{-\beta}{\beta+1}, \tfrac{1}{\beta+1}\big)\,, \quad x \mapsto T_{-\beta}^n\big((-\beta)^{-n} x\big)\,, \ \mbox{with}\ n \ge 0\ \mbox{such that}\ (-\beta)^{-n} x \in \big(\tfrac{-\beta}{\beta+1}, \tfrac{1}{\beta+1}\big).
\]
Since $T_{-\beta}\big((-\beta)^{-1} x\big) = x$ for all $x \in \big(\tfrac{-\beta}{\beta+1}, \tfrac{1}{\beta+1}\big)$, the map $\iota_\beta$ is well defined.

\begin{lemma} \label{l:Ziota}
For any $\beta > 1$, we have 
\[
\mathbb{Z}_{-\beta} = \iota_\beta^{-1}(0)\,.
\]
\end{lemma}

\begin{proof}
If $x \in \iota_\beta^{-1}(0)$, then $T_{-\beta}^n\big((-\beta)^{-n} x\big) = 0$ for some $n \ge 0$, thus $x \in (-\beta)^n\, T_{-\beta}^{-n}(0)$, i.e., $x \in \mathbb{Z}_{-\beta}$.
On the other hand, $x \in \mathbb{Z}_{-\beta}$ implies that $x \in (-\beta)^n\, T_{-\beta}^{-n}(0)$ for some $n \ge 0$. 
If $x\, (-\beta)^{-n} \in \big(\tfrac{-\beta}{\beta+1}, \tfrac{1}{\beta+1}\big)$, this immediately implies that $\iota_\beta(x) = 0$.
If $(-\beta)^{-n} x = \tfrac{-\beta}{\beta+1}$, then 
\[
T_{-\beta}^{n+2}\big((-\beta)^{-(n+2)} x\big) = T_{-\beta}^{n+2}\big(\tfrac{-1/\beta}{\beta+1}\big) = T_{-\beta}^{n+1}\big(\tfrac{-\beta}{\beta+1}\big) = T_{-\beta}^{n+1}\big((-\beta)^{-n} x\big) = T_{-\beta}(0) = 0\,,
\]
thus $\iota_\beta(x) = 0$ as well.
\end{proof}

Note that $\iota_\beta(x) = x$ for all $x \in \big(\tfrac{-\beta}{\beta+1}, \tfrac{1}{\beta+1}\big)$.
For other~$x$, we use the following relation.

\begin{lemma} \label{l:iotaT}
For any $\beta > 1$, $x \in \mathbb{R}$, we have 
\[
\iota_\beta(-\beta x) = T_{-\beta}\big(\iota_\beta(x)\big)\,.
\]
\end{lemma}

\begin{proof}
Let $x \in \mathbb{R}$ with $(-\beta)^{-n} x \in \big(\tfrac{-\beta}{\beta+1}, \tfrac{1}{\beta+1}\big)$, $n \ge 0$.
Then we have
\[
\iota_\beta(-\beta x) = T_{-\beta}^{n+1}\big((-\beta)^{-(n+1)} (-\beta x)\big) = T_{-\beta}\, T_{-\beta}^n\big((-\beta)^{-n} x\big) = T_{-\beta}\big(\iota_\beta(x)\big)\,. \qedhere
\]
\end{proof}

An important role in the study of the $(-\beta)$-transformation is played by the orbit of the left endpoint~$\tfrac{-\beta}{\beta+1}$.
In the following, fix $\beta > 1$, and let 
\[
t_n = T_{-\beta}^n\big(\tfrac{-\beta}{\beta+1}\big) \quad (n \ge 0)\,, \quad a_n  = \lfloor - t_0 - \beta t_{n-1} \rfloor \quad (n \ge 1).
\]
(As usual, $\lfloor x\rfloor$ denotes the largest integer $\le x$, and $\lceil x\rceil$ denotes the smallest integer $\ge x$.)
Then $a_1 a_2 \cdots$ is the $(-\beta)$-expansion of~$\tfrac{-\beta}{\beta+1}$, i.e., 
\[
t_n = \sum_{k=1}^\infty \frac{a_{n+k}}{(-\beta)^k} \quad \mbox{for all}\  n \ge 0,
\]
see \cite{Ito-Sadahiro09}.
Setting $t_{-1} = \tfrac{1}{\beta+1}$, $t_\infty = 0$, $\mathbb{N}_\infty = \{0,1,2,\ldots\} \cup \{\infty\}$, we consider open intervals
\[
J_{(i,j)} = (t_{2i}, t_{2j-1}) \quad \mbox{with}\ i, j \in \mathbb{N}_\infty\,,\ 0 \le t_{2i} < t_{2j-1}\ \mbox{or}\ t_{2i} < t_{2j-1} \le 0
\]
(where $2\, \infty = \infty$ and $\infty - 1 = \infty$).
We also set $a_0 = a_\infty = 0$, and 
\[
\mathcal{A} = \{(i,j) \mid i, j \in \mathbb{N}_\infty,\, 0 \le t_{2i} < t_{2j-1}\ \mbox{or}\ t_{2i} < t_{2j-1} \le 0\}\,.
\]
(Here, $(i,j)$ is a pair of elements in $\mathbb{N}_\infty$, and not an open interval.) 
Let 
\[
L_\beta\big((i,j)\big) = t_{2j-1} - t_{2i} \quad \big((i,j) \in \mathcal{A}\big)
\]
be the length of the interval $J_{(i,j)}$, and set
\[
L_\beta(v_1 \cdots v_k) = L_\beta(v_1) + \cdots + L_\beta(v_k)\,, \quad |v_1 \cdots v_k| = k\,,
\]
for any word $v_1 \cdots v_k \in \mathcal{A}^*$, where $\mathcal{A}^*$ denotes the free monoid over~$\mathcal{A}$.

Let $\psi_\beta:\, \mathcal{A}^* \to \mathcal{A}^*$ be an anti-morphism, which is defined on $(i,j) \in \mathcal{A}$ by
\[
\psi_\beta\big((i,j)\big) = (j,i+1) \quad \mbox{if}\ a_{2i+1} = a_{2j},\ t_{2i+1} t_{2j} \ge 0,
\]
and otherwise by
\[
\psi_\beta\big((i,j)\big) = \left\{\!\begin{array}{cl}
(j,\infty)\, \big((\infty,0)\, (0,\infty)\big)^{a_{2i+1}-a_{2j}}\, (\infty,i+1) & \! \mbox{if}\ t_{2i+1} > 0,\,  t_{2j} < 0, \\[1ex] 
(j,0)\, \big((0,\infty)\, (\infty,0)\big)^{a_{2i+1}-a_{2j}-1}\, (0,\infty)\, (\infty,i+1) & \! \mbox{if}\ t_{2i+1} > 0,\,  t_{2j} \ge 0, \\[1ex] 
(j,0)\, \big((0,\infty)\, (\infty,0)\big)^{a_{2i+1}-a_{2j}-1}\, (0,i+1) & \! \mbox{if}\ t_0 < t_{2i+1} \le 0,\,  t_{2j} \ge 0, \\[1ex] 
(j,\infty)\, (\infty,0)\, \big((0,\infty)\, (\infty,0)\big)^{a_{2i+1}-a_{2j}-1}\, (0,i+1) & \! \mbox{if}\ t_0 < t_{2i+1} \le 0,\,  t_{2j} < 0, \\[1ex] 
(j,0)\, \big((0,\infty)\, (\infty,0)\big)^{a_{2i+1}-a_{2j}-1} & \! \mbox{if}\ t_{2i+1} = t_0,\,  t_{2j} \ge 0, \\[1ex] 
(j,\infty)\, (\infty,0)\, \big((0,\infty)\, (\infty,0)\big)^{a_{2i+1}-a_{2j}-1} & \! \mbox{if}\ t_{2i+1} = t_0,\,  t_{2j} < 0.
\end{array}\right.
\]
Here, anti-morphism means that $\psi_\beta(v w) = \psi_\beta(w) \psi_\beta(v)$ for all $v, w \in \mathcal{A}^*$.
The anti-morphism $\psi_\beta$ is naturally extended to infinite words over~$\mathcal{A}$.
(Right infinite words are mapped to left infinite words and vice versa.)

\begin{lemma} \label{l:phiu}
Let $\beta > 1$.
For any $u \in \mathcal{A}$, we have 
\[
L_\beta\big(\psi_\beta(u)\big) = \beta L_\beta(u)\,.
\]
Moreover, for any $1 \le \ell \le |\psi_\beta(u)|$, $0 < x < L_\beta(v_\ell)$, with $\psi_\beta(u) = v_1 \cdots v_{|\psi_\beta(u)|}$, we have
\[
T_{-\beta}\big(t_{2j-1} - \beta^{-1} L_\beta\big(v_1 \cdots v_{\ell-1}\big) - \beta^{-1} x\big) = t_{2i'} + x\,, \quad \mbox{where}\ u = (i,j),\, v_\ell = (i',j').
\] 
\end{lemma}

\begin{proof}
This follows from the definitions of $T_{-\beta}$ and~$\psi_\beta$.
\end{proof}

Let $\cdots u_{-1} u_0 u_1 \cdots \in \mathcal{A}^\mathbb{Z}$ be the fixed point of $\psi_\beta$ such that $u_0 = (\infty,0)$, $u_0 u_1 \cdots$ is a fixed point of~$\psi_\beta^2$ and $\cdots u_{-2} u_{-1} = \psi_\beta(u_0 u_1 \cdots)$, in particular $u_{-1} = (0,\infty)$.
Let 
\[
Y_\beta = \{y_k \mid k \in \mathbb{Z}\} \quad \mbox{with} \quad y_k = \left\{\begin{array}{cl}L_\beta(u_0 \cdots u_{k-1}) & \mbox{if}\ k \ge 0, \\[1ex] - L_\beta(u_k \cdots u_{-1}) & \mbox{if}\ k < 0.\end{array}\right.
\]

\begin{proposition} \label{p:iotaJ}
Let $\beta > 1$.
On every interval $(y_k, y_{k+1})$, $k \in \mathbb{Z}$, the map $\iota_\beta$ is a translation, with
\[
\iota_\beta\big((y_k, y_{k+1})\big) = J_{u_k}\,. 
\]
\end{proposition}

\begin{proof}
We have $\iota_\beta(x) = x$ on the intervals $\big(y_0, y_1\big) = \big(t_\infty, t_{-1}\big) = \big(0, \tfrac{1}{\beta+1}\big)$ and $\big(y_{-1}, y_0\big) = \big(t_0, t_\infty\big) = \big(\tfrac{-\beta}{\beta+1}, 0\big)$, thus the statement of the proposition holds for $(y_k, y_{k+1}) \subset \big(\tfrac{-\beta}{\beta+1}, \tfrac{1}{\beta+1}\big)$.

Assume that the statement holds for $k \in \mathbb{Z}$.
By Lemma~\ref{l:phiu}, we have $(-\beta) \big(y_k, y_{k+1}\big) = \big(y_{k'}, y_{k'+|\psi_\beta(u_k)|}\big)$ and $\psi_\beta(u_k) = u_{k'} \cdots u_{k'+|\psi_\beta(u_k)|-1}$, with $k' = -|\psi_\beta(u_0 \cdots u_k)|$ if $k \ge 0$, $k' = |\psi_\beta(u_{k+1} \cdots u_{-1})|$ if $k < 0$.
Then Lemma~\ref{l:iotaT}, the assumption $\iota_\beta\big((y_k, y_{k+1})\big) = J_{u_k}$, and Lemma~\ref{l:phiu} yield that
\begin{align*}
\iota_\beta\big((y_{k'+\ell}, y_{k'+\ell+1})\big) & = T_{-\beta}\big(\iota_\beta\big((-\beta)^{-1} (y_{k'+\ell}, y_{k'+\ell+1})\big)\big) \\
& = T_{-\beta}\big(\iota_\beta\big(\big(y_{k+1} - \beta^{-1} L_\beta(u_{k'} \cdots u_{k'+\ell}),\, y_{k+1} - \beta^{-1} L_\beta(u_{k'} \cdots u_{k'+\ell-1})\big)\big)\big) \\
& = T_{-\beta}\big(\big(t_{2j-1} - \beta^{-1} L_\beta(u_{k'} \cdots u_{k'+\ell}),\, t_{2j-1} - \beta^{-1} L_\beta(u_{k'} \cdots u_{k'+\ell-1})\big)\big) \\
& = J_{u_{k'+\ell}}
\end{align*}
for all $0 \le \ell < |\psi_\beta(u_k)|$, where $u_k = (i,j)$.

By induction on~$n$, we obtain for every $n \ge 0$ that the statement of the proposition holds for all $k \in \mathbb{Z}$ with $(y_k, y_{k+1}) \subset (-\beta)^n \big(\tfrac{-\beta}{\beta+1}, \tfrac{1}{\beta+1}\big)$, thus it holds for all $k \in \mathbb{Z}$.
\end{proof}

Now we describe the set $Y_\beta = \{y_k \mid k \in \mathbb{Z}\}$, which is left out by the intervals $(y_k, y_{k+1})$. 

\begin{lemma} \label{l:Y}
For any $\beta > 1$, we have
\[
Y_\beta =\, \mathbb{Z}_{-\beta}\, \cup \bigcup_{m,n\ge0} (-\beta)^{m+n}\, T_{-\beta}^{-n}\big(\tfrac{-\beta}{\beta+1}\big)\,.
\]
\end{lemma}

\begin{proof}
First note that 
\[
\bigcup_{m,n\ge0} (-\beta)^{m+n}\, T_{-\beta}^{-n}\big(\tfrac{-\beta}{\beta+1}\big) = \bigcup_{m\ge0} (-\beta)^m\, \iota_\beta^{-1}\big(\tfrac{-\beta}{\beta+1}\big)\,,
\]
similarly to Lemma~\ref{l:Ziota}.
Indeed, $x \in \iota_\beta^{-1}(t_0)$ is equivalent to $(-\beta)^{-n} x \in T_{-\beta}^{-n}(t_0) \cap (t_0,t_{-1})$ for some $n \ge 0$.
In the remaining case $(-\beta)^{-n} x = t_0 \in T_{-\beta}^{-n}(t_0)$, we have $x \in (-\beta)\, \iota_\beta^{-1}(t_0)$ since $T_{-\beta}^{n+1}\big((-\beta)^{-(n+2)} x\big) = T_{-\beta}^n(t_0) = t_0$.
Note also that $(-\beta)^m\, \iota_\beta^{-1}(t_0) \subseteq \iota_\beta^{-1}(t_m)$.

Since $t_0 \not\in J_{u_k}$ and $0 \not\in J_{u_k}$ for all $k \in \mathbb{Z}$ by the definition of~$\mathcal{A}$, Proposition~\ref{p:iotaJ} implies that $\iota_\beta^{-1}(0) \cup \iota_\beta^{-1}(t_0) \subseteq Y_\beta$, thus $\mathbb{Z}_{-\beta} \subseteq Y_\beta$ by Lemma~\ref{l:Ziota}.
Since $(-\beta) Y_\beta \subseteq Y_\beta$ by Lemma~\ref{l:phiu}, we obtain $(-\beta)^m\, \iota_\beta^{-1}(t_0) \subseteq Y_\beta$ for every $m \ge 0$ as well.

Let now $x \in Y_\beta \setminus \{0\}$.
Then there exists some $m \ge 0$ such that $(-\beta)^{-m} x \in Y_\beta$ and $(-\beta)^{-m-1} x \not\in Y_\beta$, i.e., $(-\beta)^{-m-1} x \in (y_k, y_{k+1})$ for some $k \in \mathbb{Z}$.
By Lemma~\ref{l:iotaT} and the definition of~$\psi_\beta$, we obtain that $\iota_\beta\big((-\beta)^{-m} x\big) \in \{t_0, 0\}$, i.e., $x \in (-\beta)^m\, \iota_\beta^{-1}(t_0)$ or $x \in (-\beta)^m\, \mathbb{Z}_{-\beta} \subseteq \mathbb{Z}_{-\beta}$.
Since $0 \in \mathbb{Z}_{-\beta}$, this proves the lemma.
\end{proof}

\begin{theorem} \label{t:Zu}
Let $\beta > 1$, $x \in \mathbb{R}$.
Then $x \in \mathbb{Z}_{-\beta}$ if and only if $x = y_k$ for some $k \in \mathbb{Z}$ with $t_{2j-1} = 0$ or $t_{2i} = 0$, where $u_{k-1} = (i',j)$, $u_k = (i,j')$.
\end{theorem}

\begin{proof}
Since $\mathbb{Z}_{-\beta} \subseteq Y_\beta$ by Lemma~\ref{l:Y}, it is sufficient to consider $x = y_k$, $k \in \mathbb{Z}$.
As in the proof of Lemma~\ref{l:Y}, let $m \ge 0$ be such that $(-\beta)^{-m} x \in Y_\beta$ and $(-\beta)^{-m-1} x \not\in Y_\beta$.
Then we have $\iota_\beta\big((-\beta)^{-m} x\big) \in \{t_0, 0\}$.

If $\iota_\beta\big((-\beta)^{-m} x\big) = 0$, then $\iota_\beta(x) = T_{-\beta}^m(0) = 0$.
Moreover, $\iota_\beta$ is continuous at $(-\beta)^{-m} x$ in this case. 
Together with the continuity of $T_{-\beta}^m$ at~$0$, this implies that $\iota_\beta$ is continuous at~$x$, i.e., $u_{k-1} = (i',j)$ with $t_{2j-1} = 0$ and $u_k = (i,j')$ with $t_{2i} = 0$.

If $\iota_\beta\big((-\beta)^{-m} x\big) = t_0$, then $\iota_\beta$ is right continuous at $(-\beta)^{-m} x$, and its limit from the left is~$t_{-1}$. 
We obtain that $\iota_\beta(x) = t_m$, $u_{k-1} = (i',\lceil m/2 \rceil)$ and $u_k = (\lfloor m/2 \rfloor,j')$ for some $i',j' \in \mathbb{N}_\infty$.
Let $i = \lfloor m/2 \rfloor$, $j = \lceil m/2 \rceil$.
Since $2i = m$ if $m$ is even and $2j-1 = m$ if $m$ is odd, $x \in \mathbb{Z}_{-\beta}$ implies that $t_{2i} = 0$ or $t_{2j-1} = 0$. 
On the other hand, $t_{2i} = 0$ or $t_{2j-1} = 0$ yields that $t_{m-1} = 0$ or $t_m = 0$. 
Since $t_{m-1} = 0$ implies $t_m = 0$, we must have $x \in \mathbb{Z}_{-\beta}$.
\end{proof}

By Theorem~\ref{t:Zu}, the study of $\mathbb{Z}_{-\beta}$ is reduced to the study of the fixed point of~$\psi_\beta$.
Note that $(i,j)$ and $(i',j')$ can be identified when $t_{2i} = t_{2i'}$ and $t_{2j-1} = t_{2j'-1}$.
After identification, $\mathcal{A}$~is finite if and only if $\{t_n \mid n \ge 0\}$ is a finite set, i.e., if $\beta$ is an Yrrap number. 

With the help of~$\psi_\beta$, we can construct an anti-morphism describing the structure of~$\mathbb{Z}_{-\beta}$ (for $\beta \ge \frac{1+\sqrt{5}}{2}$), similarly to \cite{Steiner}.
First note that $\beta \ge \frac{1+\sqrt{5}}{2}$ implies $a_1 = 1$, $t_1 \ge 0$, or $a_1 \ge 2$, thus 
$\psi_\beta^2\big((\infty,0)\big) = \psi_\beta\big((0,\infty)\big)$ starts with $(\infty,0)\, (0,\infty)$.
Therefore, $L_\beta\big((\infty,0)\, (0,\infty)\big) = t_{-1} - t_0 = 1$ is the smallest positive element of~$\mathbb{Z}_{-\beta}$.
Using Theorem~\ref{t:Zu} and $(-\beta) \mathbb{Z}_{-\beta} \subseteq \mathbb{Z}_{-\beta}$, the word $\psi_\beta\big((\infty,0)\, (0,\infty)\big)$ determines the set $\mathbb{Z}_{-\beta} \cap [-\beta,0]$.
Splitting up $\psi_\beta\big((\infty,0)\, (0,\infty)\big)$ according to Theorem~\ref{t:Zu} and applying $\psi_\beta$ on each of the factors yields the set $\mathbb{Z}_{-\beta} \cap [0,\beta^2]$, etc.
Consider all these factors, i.e., all words $u_k \cdots u_{k'-1}$ between consecutive elements $y_k, y_{k'} \in \mathbb{Z}_{-\beta}$, as letters.
Using the described strategy, we define an anti-morphism $\varphi_{-\beta}$ on words consisting of these letters.
Then the fixed point of $\varphi_{-\beta}$ codes the distances between the elements of~$\mathbb{Z}_{-\beta}$, see the two simple examples below. 
For more complicated examples, we refer to~\cite{Steiner}.
By \cite{Ambroz-Dombek-Masakova-Pelantova,Steiner}, the alphabet is finite if and only if $\beta$ is an Yrrap number.

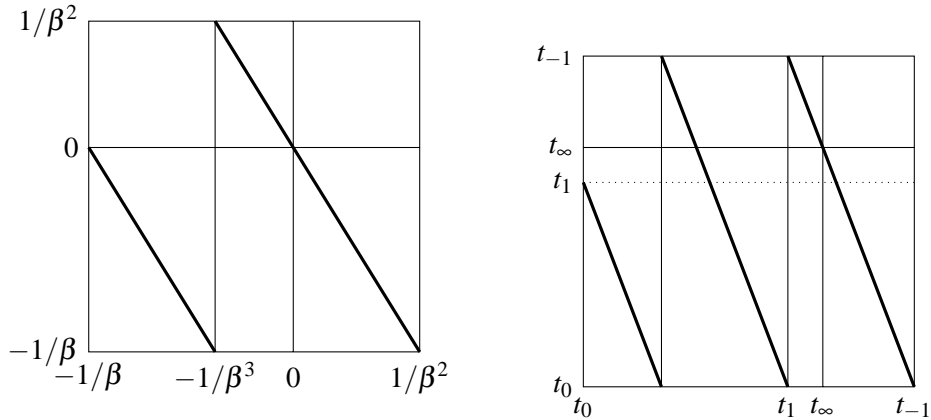
\begin{figure}[ht]
\centering
\begin{tikzpicture}[scale=4.4]
\draw(-.618,0)node[left]{$0$}--(.382,0) (0,-.618)node[below]{$\vphantom{\beta^3}0$}--(0,.382)
(-.618,-.618)node[below]{$-1/\beta$}node[left]{$-1/\beta$}--(.382,-.618)node[below]{$1/\beta^2$}--(.382,.382)--(-.618,.382)node[left]{$1/\beta^2$}--cycle
(-.2361,-.618)node[below]{$-1/\beta^3$}--(-.2361,.382);
\draw[very thick](-.618,0)--(-.2361,-.618) (-.2361,.382)--(.382,-.618);

\begin{scope}[shift={(1.6,0)}]
\draw(-.7236,0)node[left]{$t_\infty$}--(.2764,0) (0,-.7236)node[below]{$t_\infty$}--(0,.2764)
(-.7236,-.7236)node[below]{$t_0$}node[left]{$t_0$}--(.2764,-.7236)node[below]{$t_{-1}$}--(.2764,.2764)--(-.7236,.2764)node[left]{$t_{-1}$}--cycle
(-.4875,-.7236)--(-.4875,.2764) (-.1056,-.7236)node[below]{$t_1$}--(-.1056,.2764);
\draw[very thick](-.7236,-.1056)--(-.4875,-.7236) (-.4875,.2764)--(-.1056,-.7236) (-.1056,.2764)--(.2764,-.7236);
\draw[dotted](-.7236,-.1056)node[left]{$t_1$}--(.2764,-.1056);
\end{scope}
\end{tikzpicture}
\caption{The $(-\beta)$-transformation for $\beta = \frac{1+\sqrt5}{2}$ (left) and $\beta = \frac{3+\sqrt5}{2}$ (right).}
\label{f:Tmbeta}
\end{figure}

\begin{example} \label{x:gm2}
Let $\beta = \frac{3+\sqrt5}{2} \approx 2.618$, i.e., $\beta^2 = 3 \beta - 1$.
Then we have $t_1 = \frac{\beta^2}{\beta+1} - 2 = \frac{-1/\beta}{\beta+1}$ and $t_2 = \frac{1}{\beta+1} - 1 = t_0$.
Therefore, we can identify $(0,\infty)$ and $(1,\infty)$, and obtain 
\begin{align*}
\psi_\beta: \quad (\infty,0) & \mapsto (0,\infty)\,, \\
(0,\infty) & \mapsto (\infty,0)\, (0,\infty)\, (\infty,0)\, (0,1)\,, \\
(0,1) & \mapsto (0,\infty)\, (\infty,0)\, (0,1)\,.
\end{align*}
The two-sided fixed point $\cdots u_{-2} u_{-1} \dot{~} u_0 u_1 \cdots$ of $\psi_\beta$ is equal to
\[
\cdots (0,\infty) (\infty,0) (0,1)\, (0,\infty)\, (\infty,0) (0,\infty) (\infty,0)\, (0,1)\, (0,\infty) \dot{~} (\infty,0) (0,\infty) (\infty,0) (0,1) \cdots\,.
\]
Applying $\psi_\beta$ to $(\infty,0)\, (0,\infty)$ and to the factors described by Theorem~\ref{t:Zu} yields 
\begin{align*}
(\infty,0)\, (0,\infty) & \mapsto (\infty,0)\, (0,\infty)\ (\infty,0)\, (0,1)\, (0,\infty)\,, \\
(\infty,0)\, (0,1)\, (0,\infty) & \mapsto (\infty,0)\, (0,\infty)\ (\infty,0)\, (0,1)\, (0,\infty)\ (\infty,0)\, (0,1)\, (0,\infty)\,.
\end{align*}
Therefore, setting $A = (\infty,0)\, (0,\infty)$ and $B = (\infty,0)\, (0,1)\, (0,\infty)$, the fixed point 
\[
\cdots A B B \, A B \, A B B \, A B B \, A B \,\dot{~} A B B \, A B \, A B B \, A B B \, A B \, \cdots
\]
of the anti-morphism
\[
\varphi_{-\beta}: \quad A \mapsto A B \,,  \quad B \mapsto A B B\,, 
\]
describes the set of $(-\beta)$-integers, see Figure~\ref{f:intgm2}.
The distances between consecutive elements of $\mathbb{Z}_{-\beta}$ are $L_\beta(A) = 1$ and $L_\beta(B) = \beta - 1 > 1$.
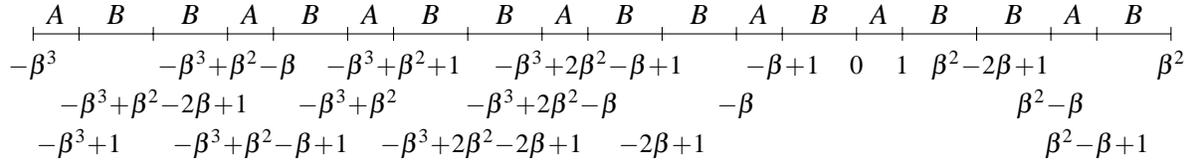
\begin{figure}[ht]
\centering
\begin{tikzpicture}[scale=.61]
\draw(-17.944,-.15)node[below]{\small$-\beta^3$}--(-17.944,.15) (-16.944,-.15)node[below=6ex]{\small$-\beta^3\!+\!1$}--(-16.944,.15) (-15.326,-.15)node[below=3ex]{\small$-\beta^3\!+\!\beta^2\!-\!2\beta\!+\!1$}--(-15.326,.15) (-13.708,-.15)node[below]{\small$-\beta^3\!+\!\beta^2\!-\!\beta$}--(-13.708,.15)  (-12.708,-.15)node[below=6ex]{\small$\hspace{-1em}-\beta^3\!+\!\beta^2\!-\!\beta\!+\!1$}--(-12.708,.15) (-11.09,-.15)node[below=3ex]{\small$-\beta^3\!+\!\beta^2$}--(-11.09,.15) (-10.09,-.15)node[below]{\small$-\beta^3\!+\!\beta^2\!+\!1$}--(-10.09,.15) (-8.472,-.15)node[below=6ex]{\small$\hspace{1em}-\beta^3\!+\!2\beta^2\!-\!2\beta\!+\!1$}--(-8.472,.15) (-6.854,-.15)node[below=3ex]{\small$-\beta^3\!+\!2\beta^2\!-\!\beta$}--(-6.854,.15) (-5.854,-.15)node[below]{\small$-\beta^3\!+\!2\beta^2\!-\!\beta\!+\!1$}--(-5.854,.15) (-4.236,-.15)node[below=6ex]{\small$-2\beta\!+\!1\vphantom{\beta^2}$}--(-4.236,.15) (-2.618,-.15)node[below=3ex]{\small$-\beta\vphantom{\beta^2}$}--(-2.618,.15) (-1.618,-.15)node[below]{\small$-\beta\!+\!1\vphantom{\beta^2}$}--(-1.618,.15) (0,-.15)node[below]{\small$0\vphantom{\beta^2}$}--(0,.15) (1,-.15)node[below]{\small$1\vphantom{\beta^2}$}--(1,.15) (2.618,-.15)node[below]{\small$\hspace{1em}\beta^2\!-\!2\beta\!+\!1$}--(2.618,.15) (4.236,-.15)node[below=3ex]{\small$\beta^2\!-\!\beta$}--(4.236,.15) (5.236,-.15)node[below=6ex]{\small$\beta^2\!-\!\beta\!+\!1$}--(5.236,.15) (6.854,-.15)node[below]{\small$\beta^2$}--(6.854,.15) 
(-17.944,0)--node[above]{$A$}(-16.944,0)--node[above]{$B$}(-15.326,0)--node[above]{$B$}(-13.708,0)--node[above]{$A$}(-12.708,0)--node[above]{$B$}(-11.09,0)--node[above]{$A$}(-10.09,0)--node[above]{$B$}(-8.472,0)--node[above]{$B$}(-6.854,0)--node[above]{$A$}(-5.854,0)--node[above]{$B$}(-4.236,0)--node[above]{$B$}(-2.618,0)--node[above]{$A$}(-1.618,0)--node[above]{$B$}(0,0)--node[above]{$A$}(1,0)--node[above]{$B$}(2.618,0)--node[above]{$B$}(4.236,0)--node[above]{$A$}(5.236,0)--node[above]{$B$}(6.854,0);
\end{tikzpicture}
\caption{The set $\mathbb{Z}_{-\beta} \cap [-\beta^3, \beta^2]$, $\beta = (3+\sqrt5)/2$.}
\label{f:intgm2}
\end{figure}
\end{example}

\begin{example} \label{x:gm}
Let $\beta = \frac{1+\sqrt5}{2} \approx 1.618$, i.e., $\beta^2 = \beta + 1$.
Then we have $t_1 = \frac{\beta^2}{\beta+1} - 1 = 0 = t_\infty$.
Identifying $(0,\infty)$ and $(0,1)$, we obtain 
\[
\psi_\beta: \quad (\infty,0) \mapsto (0,\infty)\,, \quad (0,\infty) \mapsto (\infty,0)\, (0,\infty)\,.
\]
The two-sided fixed point $\cdots u_{-2} u_{-1} \dot{~} u_0 u_1 \cdots$ of $\psi_\beta$ is equal to
\[
\cdots (\infty,0) (0,\infty)\, (0,\infty)\, \dot{~} (\infty,0) (0,\infty)\, (\infty,0) (0,\infty)\, (0,\infty)\, \cdots\,.
\]
The words $u_k \cdots u_{k'-1}$ between consecutive elements $y_k, y_{k'} \in \mathbb{Z}_{-\beta}$ are $A = (\infty,0)\, (0,\infty)$ and $B = (0,\infty)$, since
\[
(\infty,0)\, (0,\infty) \mapsto (\infty,0)\, (0,\infty)\ (0,\infty)\,, \quad (0,\infty) \mapsto (\infty,0)\, (0,\infty)\,.
\]
Note that $B$ does not start with a letter $(i,j)$ with $t_{2i} = 0$, thus $\iota_\beta$ is discontinuous at the corresponding points $y_k \in \mathbb{Z}_{-\beta}$.
The fixed point 
\[
\cdots A \, A B \, A \, A B \, A B \, A \, A B \, A B \,\dot{~} A \, A B \, A \, A B \, A B \, A \, A B \, A \, A B \, A B \, A \, A B \, A B \, \cdots
\]
of the anti-morphism
\[
\varphi_{-\beta}: \quad A \mapsto A B \,,  \quad B \mapsto A \,, 
\]
describes the set of $(-\beta)$-integers, with $L_\beta(A) = 1$ and $L_\beta(B) = \beta - 1 < 1$, see Figure~\ref{f:intgm}.
Note that $(-\beta)^n$ can also be represented as $(-\beta)^{n+2} + (-\beta)^{n+1}$.
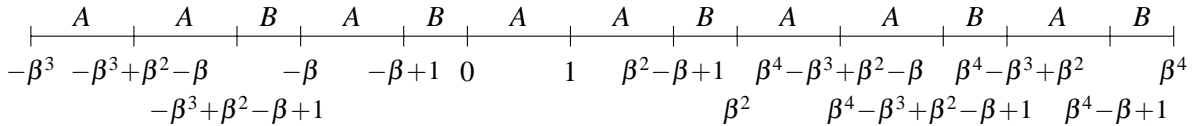
\begin{figure}[ht]
\centering
\begin{tikzpicture}[scale=1.37]
\draw(-4.236,-.1)node[below]{\small$-\beta^3$}--(-4.236,.1) (-3.236,-.1)node[below]{\small$\hspace{.5em}-\beta^3\!+\!\beta^2\!-\!\beta$}--(-3.236,.1) (-2.236,-.1)node[below=3ex]{\small$-\beta^3\!+\!\beta^2\!-\!\beta\!+\!1$}--(-2.236,.1) (-1.618,-.1)node[below]{\small$-\beta\vphantom{\beta^2}$}--(-1.618,.1) (-.618,-.1)node[below]{\small$-\beta\!+\!1\vphantom{\beta^2}$}--(-.618,.1) (0,-.1)node[below]{$0\vphantom{\beta^2}$}--(0,.1) (1,-.1)node[below]{$1\vphantom{\beta^2}$}--(1,.1) (2,-.1)node[below]{\small$\beta^2\!-\!\beta\!+\!1$}--(2,.1) (2.618,-.1)node[below=3ex]{\small$\beta^2$}--(2.618,.1) (3.618,-.1)node[below]{\small$\beta^4\!-\!\beta^3\!+\!\beta^2\!-\!\beta$}--(3.618,.1) (4.618,-.1)node[below=3ex]{\small$\hspace{-1em}\beta^4\!-\!\beta^3\!+\!\beta^2\!-\!\beta\!+\!1$}--(4.618,.1) (5.236,-.1)node[below]{\small$\hspace{1em}\beta^4\!-\!\beta^3\!+\!\beta^2$}--(5.236,.1) (6.236,-.1)node[below=3ex]{\small$\hspace{.5em}\beta^4\!-\!\beta\!+\!1$}--(6.236,.1) (6.854,-.1)node[below]{\small$\beta^4$}--(6.854,.1) 
 (-4.236,0)--node[above]{$A$}(-3.236,0)--node[above]{$A$} (-2.236,0)--node[above]{$B$}(-1.618,0)--node[above]{$A$}(-.618,0)--node[above]{$B$}(0,0)--node[above]{$A$}(1,0)--node[above]{$A$}(2,0)--node[above]{$B$}(2.618,0)--node[above]{$A$}(3.618,0)--node[above]{$A$}(4.618,0)--node[above]{$B$}(5.236,0)--node[above]{$A$}(6.236,0)--node[above]{$B$}(6.854,0);
\end{tikzpicture}
\caption{The $(-\beta)$-integers in $[-\beta^3,\beta^4]$, $\beta = (1+\sqrt5)/2$.}
\label{f:intgm}
\end{figure}
\end{example}

\section{Delone property} \label{sec:delone-property}
A~set $S \subset \mathbb{R}$ is called \emph{Delone set} (or Delaunay set) if it is \emph{uniformly discrete} and \emph{relatively dense}; i.e., if there are numbers $R>r>0$, such that each interval of length~$r$ contains at most one point of~$S$, and every interval of length~$R$ contains at least one point of~$S$. 

If $\beta$ is an Yrrap number, then the set of distances between consecutive $(-\beta)$-integers is finite by \cite{Ambroz-Dombek-Masakova-Pelantova,Steiner}, thus $\mathbb{Z}_{-\beta}$ is a Delone set. 
For general~$\beta$, we show in this section that $\mathbb{Z}_{-\beta}$ need neither be uniformly discrete nor relatively dense.

\subsection{Uniform discreteness}
Since every point $x \in Y_\beta$ is separated from $Y_\beta \setminus \{x\}$ by an interval around~$x$, $Y_\beta$~is discrete, and the same holds for~$\mathbb{Z}_{-\beta}$.
It is well known that $\mathbb{Z}_\beta$ is uniformly discrete if and only if $0$ is not an accumulation point of $\{T_\beta^n(1) \mid n \ge 0\}$ (where $T_\beta(1) = \beta-\lfloor\beta\rfloor$).
For $(-\beta)$-integers, the situation is more complicated.

\begin{proposition} \label{p:ud}
Let $\beta > 1$.
If $0$ is not an accumulation point of $\big\{T_{-\beta}^{2n-1}\big(\frac{-\beta}{\beta+1}\big) > 0 \mid n \ge 1\big\}$, then the set $\mathbb{Z}_{-\beta}$ is uniformly discrete.
\end{proposition}

\begin{proof}
When $\beta$ is an Yrrap number, then $\{L_\beta(u_k) \mid k \in \mathbb{Z}\}$ is finite, thus $\mathbb{Z}_{-\beta}$ is uniformly discrete.
Therefore, assume that $\beta$ is not an Yrrap number, in particular that $t_n \ne 0$ for all (finite) $n \ge 0$, with the notation of Section~\ref{sec:fixed-point-an}.
Then $\iota_\beta$ is continuous at every point $y_k \in \mathbb{Z}_{-\beta}$, $k \in \mathbb{Z}$, thus $u_k = (i,j)$ with $t_{2i} = 0$. 
Since $L_\beta(u_k) \ge \inf\{t_{2n-1} > 0 \mid n \ge 0\}$, the set $\mathbb{Z}_{-\beta}$~is uniformly discrete if $\inf\{t_{2n-1} > 0 \mid n \ge 1\} > 0$.
\end{proof}

In order to give examples of~$\beta$ where $\mathbb{Z}_{-\beta}$ is not uniformly discrete or not relatively dense, we need to know which sequences are possible $(-\beta)$-expansions of $\frac{-\beta}{\beta+1}$.
The corresponding problem for $\beta$-expansions was solved by Parry \cite{Parry60}. 
G\'ora \cite[Theorem~25]{Gora07} gave an answer to a more general question, but his theorem is incorrect, as noticed in \cite{Dombek-Masakova-Pelantova}.
However, G\'ora proved the following result, where $\le_{\mathrm{alt}}$ denotes the alternate order on words, i.e., $x_1 x_2 \cdots <_{\mathrm{alt}} y_1 y_2 \cdots$ if $(-1)^n (x_{n+1}-y_{n+1}) < 0$, $x_1 \cdots x_n = y_1 \cdots y_n$ for some $n \ge 0$.

\begin{lemma} \label{l:gora}
Let $a_1 a_2 \cdots$ be a sequence of non-negative integers satisfying $a_{n+1} a_{n+2} \cdots \le_{\mathrm{alt}} a_1 a_2 \cdots$ for all $n \ge 1$, with $a_1 \ge 2$.
Then there exists a unique $a_1 \le \beta \le a_1+1$ such that 
\[
\sum_{k=1}^\infty \frac{a_k}{(-\beta)^k} = \frac{-\beta}{\beta+1}\,, \quad \mbox{and} \quad \frac{-\beta}{\beta+1} \le \sum_{k=1}^\infty \frac{a_{n+k}}{(-\beta)^k} \le \frac{1}{\beta+1} \quad \mbox{for all}\ n \ge 1\,.
\]
\end{lemma}

\begin{proof}
As the generalised $\beta$-transformations in \cite{Gora07} with $E = (1,\ldots,1)$ are intimately related to our~$T_{-\beta}$ (see the introduction of \cite{Liao-Steiner}), the lemma follows from the proof of Theorem~25 in \cite{Gora07}.
\end{proof}

The flaw in \cite[Theorem~25]{Gora07} is that $\sum_{k=1}^\infty \frac{a_{n+k}}{(-\beta)^k} = \frac{1}{\beta+1}$ cannot be excluded, which would be necessary for $a_1 a_2 \cdots$ to be the $(-\beta)$-expansion of~$\frac{-\beta}{\beta+1}$.
Indeed, the sequence $a_1 a_2 \cdots = 2 (10)^\omega$ satisfies the conditions of Lemma~\ref{l:gora} and yields $\beta = 2$, but $\sum_{k=1}^\infty \frac{a_{k+2}}{(-2)^k} = \frac{1}{3}$; the $(-2)$-expansion of $\frac{-2}{3}$ is $2^\omega$.
In order to avoid this problem, we define a relation $<'_{\mathrm{alt}}$ by $x_1 x_2 \cdots <'_{\mathrm{alt}} y_1 y_2 \cdots$ if $(-1)^n (x_{n+1}-y_{n+1}) < -1$, $x_1 \cdots x_n = y_1 \cdots y_n$ for some $n \ge 0$.

\begin{lemma} \label{l:exp}
Let $a_1 a_2 \cdots$ and $\beta$ be as in Lemma~\ref{l:gora}.
If $a_{n+1} a_{n+2} \cdots <'_{\mathrm{alt}} a_1 a_2 \cdots$ for all $n \ge 2$ such that $a_n = 0$, then $a_1 a_2 \cdots$ is the $(-\beta)$-expansion of~$\frac{-\beta}{\beta+1}$.
\end{lemma}

\begin{proof}
We have to show that $\sum_{k=1}^\infty \frac{a_{n-1+k}}{(-\beta)^k} < \frac{1}{\beta+1}$ for all $n \ge 2$.
If $a_n > 0$, then this inequality follows from $\sum_{k=1}^\infty \frac{a_{n-1+k}}{(-\beta)^k} \le \frac{-1}{\beta} + \frac{1}{\beta+1} < 0$.
If $a_n = 0$, then we have some $m \ge 0$ such that $(-1)^m (a_{n+m+1}-a_{m+1}) \le -2$, $a_{n+1} \cdots a_{n+m} = a_1 \cdots a_m$.
This implies that 
\begin{multline*}
\sum_{k=1}^\infty \frac{a_{n+k}}{(-\beta)^k} = \sum_{k=1}^m \frac{a_k}{(-\beta)^k} + \frac{1}{(-\beta)^m} \sum_{k=1}^\infty \frac{a_{n+m+k}}{(-\beta)^k} \\
\ge \sum_{k=1}^m \frac{a_k}{(-\beta)^k} + \frac{a_{m+1}}{(-\beta)^{m+1}} + \frac{2-\frac{\beta}{\beta+1}}{\beta^{m+1}} > \sum_{k=1}^{m+1} \frac{a_k}{(-\beta)^k} + \frac{\frac{\beta}{\beta+1}}{\beta^{m+1}} \ge \sum_{k=1}^\infty \frac{a_k}{(-\beta)^k} = \frac{-\beta}{\beta+1}\,,
\end{multline*}
thus $\sum_{k=1}^\infty \frac{a_{n-1+k}}{(-\beta)^k} = \frac{1}{-\beta} \sum_{k=1}^\infty \frac{a_{n+k}}{(-\beta)^k} < \frac{1}{\beta+1}$.
\end{proof}

\begin{proposition}
Let $a_1 a_2 \cdots = 3 0 1 0^3 1 0^5 1 \cdots$. 
Then $a_1 a_2 \cdots$ is the $(-\beta)$-expansion of~$\frac{-\beta}{\beta+1}$ for some $\beta > 1$, and $\mathbb{Z}_{-\beta}$ is not uniformly discrete.
\end{proposition}

\begin{proof}
By Lemmas~\ref{l:gora} and~\ref{l:exp}, there exists a $\beta > 1$ such that $a_1 a_2 \cdots$ is the $(-\beta)$-expansion of~$\frac{-\beta}{\beta+1}$. 
Since $a_{2n} a_{2n+1} \cdots$ starts with an odd number of zeros for all $n \ge 1$, we have $t_{2n-1} > 0$ for all $n \ge 0$.
Therefore, induction on~$n$ yields that $\psi_\beta^{2n}\big((\infty,0)\big)$ ends with $(\infty,n)$, and $\psi_\beta^{2n+1}\big((\infty,0)\big)$ starts with $(n,\infty)$ for all $n \ge 0$.
This implies that $\psi_\beta^{2n+1}\big((\infty,0)\, (0,\infty)\big)$ contains the factor $(\infty,n)\, (n,\infty)$ for all $n \ge 0$, and $L_\beta\big((\infty,n)\, (n,\infty)\big) = t_{2n-1} - t_{2n}$ is a distance between consecutive $(-\beta)$-integers.
For any $k \ge 1$, we have $a_{k(k-1)+2} \cdots a_{k(k+1)+1} = 0^{2k-1}1$, thus $0< t_{k(k-1)+1} < \beta^{-2k} \frac{\beta^2}{\beta+1}$ and $-\beta^{-2k} \frac{\beta^3}{\beta+1} < t_{k(k-1)+2} < 0$.
Therefore, the distance between consecutive elements of $\mathbb{Z}_{-\beta}$ can be arbitrarily small. 
\end{proof}

The following proposition shows that the converse of Proposition~\ref{p:ud} is not true.

\begin{proposition}
Let $a_1 a_2 \cdots = 30032\, 00\, 00\, 30032\, 2\, 00\, 00\, 00\, 00\, 30032\, 00\, 00\, 30032\, 2\, 2 \cdots$ be a fixed point of the morphism
\[
\sigma_1:\ 3 \mapsto 30032\,,\ 2 \mapsto 2\,,\ 0 \mapsto 00\,.
\]
Then $a_1 a_2 \cdots$ is the $(-\beta)$-expansion of~$\frac{-\beta}{\beta+1}$ for a $\beta > 1$, $\inf\big\{T_{-\beta}^{2n-1}\big(\frac{-\beta}{\beta+1}\big) > 0 \mid n \ge 1\big\} = 0$ and $\mathbb{Z}_{-\beta}$ is uniformly discrete.
\end{proposition}

\begin{proof}
When $a_{n+1} = 3$, $n \ge 1$, then $a_{n+1} a_{n+2} \cdots$ starts with $\sigma_1^k(3)\, 2$ for some $k \ge 0$.
Since $|\sigma_1^k(3)|$ is odd for all $k \ge 0$, we have $\sigma_1^k(3)\, 2 <'_{\mathrm{alt}} \sigma_1^k(3)\, 0$, thus $a_{n+1} a_{n+2} \cdots <'_{\mathrm{alt}} a_1 a_2 \cdots$.
Since $a_n =0$ implies $a_{n+1} \in \{0,3\}$, the conditions of Lemmas~\ref{l:gora} and~\ref{l:exp} are satisfied, i.e., there exists a $\beta > 1$ such that $a_1 a_2 \cdots$ is the $(-\beta)$-expansion of~$\frac{-\beta}{\beta+1}$. 

Let $n_k =  |\sigma_1^k(300)| + |\sigma_1^{k-1}(3)|+1$, $k \ge 1$.
Then $a_{n_k+1} \cdots a_{n_k+2^k} = 0^{2^k-1}\,3$, thus we have $t_{n_k} > 0$ and $\lim_{k\to\infty} t_{n_k} = 0$.
Since $n_k$ is odd, this yields that $\inf\{t_{2n-1} > 0 \mid n \ge 1\} = 0$.

Let $n \ge 1$ with $t_{2n-1} > 0$.
Then $a_{2n} a_{2n+1} \cdots$ starts with an odd number of zeros, thus $a_1 \cdots a_{2n-1}$ ends with $\sigma_1^k(3)\, 0^{2j-1}$ for some $k \ge 0$, $1 \le j \le 2^{k-1}$.
Let $2m = |\sigma_1^k(3)| +  2j - 1$.
Recall that $\cdots u_{-1} u_0 u_1 \cdots$ is a fixed point of~$\psi_\beta$.
Any letter $u_k = (i,n)$, $i \in \mathbb{N}_\infty$, occurs only in $\psi_\beta^{2m}\big((i',n-m)\big)$, $i' \in \mathbb{N}_\infty$ with $t_0 \le t_{2i'} < t_{2n-2m-1}$.
Since $a_1 \cdots a_{2m} = a_{2n-2m} \cdots a_{2n-1}$, we obtain that $t_{2m} \le t_{2i} < t_{2n-1}$.
Moreover, $a_{2m+1} a_{2m+2} \cdots$ starts with $0^{2^k-2j+1}\, 3$, thus $t_{2m} > 0$.
Now, the continuity of $\iota_\beta$ at every point in~$\mathbb{Z}_{-\beta}$ yields that $y_k \not\in \mathbb{Z}_{-\beta}$ if $u_k = (i,n)$, $n \ge 1$.
Therefore, $y_k \in \mathbb{Z}_{-\beta}$ implies that $u_k = (\infty,0)$, hence $\mathbb{Z}_{-\beta}$ is uniformly discrete.
\end{proof}

\subsection{Relative denseness}
Since the distance between consecutive $\beta$-integers is at most~$1$, $\mathbb{Z}_\beta$~is always relatively dense.
We show that this is not always true for~$\mathbb{Z}_{-\beta}$.

\begin{proposition}
Let $a_1 a_2 \cdots = 31232 1 2 31232 2 1 2 31232 1 2 31232 2 2 \cdots$ be a fixed point of the morphism
\[
\sigma_2:\ 3 \mapsto 31232\,,\ 2 \mapsto 2\,,\ 1 \mapsto 1\,.
\]
Then $a_1 a_2 \cdots$ is the $(-\beta)$-expansion of~$\frac{-\beta}{\beta+1}$ for a $\beta > 1$, and $\mathbb{Z}_{-\beta}$ is not relatively dense.
\end{proposition}

\begin{proof}
Since $|\sigma_2^k(3)|$ is odd for all $k \ge 0$, we have $\sigma_2^k(3)\, 2 <_{\mathrm{alt}} \sigma_2^k(3)\, 1$, thus $a_1 a_2 \cdots$ satisfies the conditions of Lemma~\ref{l:gora}.
Since $a_n > 0$ for all $n \ge 2$, the condition of Lemma~\ref{l:exp} holds trivially.
Therefore, there exists a $\beta > 1$ such that $a_1 a_2 \cdots$ is the $(-\beta)$-expansion of~$\frac{-\beta}{\beta+1}$. 

Next we show that, for any $k \ge 0$, $\psi_\beta^{|\sigma_2^k(3)|}\big((\infty,0)\big)$ starts with
\begin{equation} \label{e:psisigma}
\big(\big\lfloor|\sigma_2^k(3)|/2\big\rfloor, \big\lceil|\sigma_2^{k-1}(3)|/2\big\rceil\big) \cdots \big(\big\lfloor|\sigma_2^1(3)|/2\big\rfloor, \big\lceil|\sigma_2^0(3)|/2\big\rceil\big)\, \big(0,\infty\big)\,.
\end{equation}
We have $\psi_\beta\big((\infty,0)\big) = (0,\infty)$ and
\[
\psi_\beta^4\big((0,\infty)\big) = \cdots \psi_\beta^3\big((0,1)\big) = \psi_\beta^2\big((1,\infty)\big) \cdots = \cdots \psi_\beta\big((\infty,0)\, (0,2)\big) = (2,1)\, (0,\infty) \cdots\,,
\]
where we have used that $a_1 > a_2$, $a_4 = a_1$, and $a_n > 0$ for all $n \ge 1$, which implies $t_n < 0$ for all $n \ge 1$.
This yields the statement for $k = 0$ and $k = 1$.
Supose that $\psi_\beta^{|\sigma_2^k(3)|}\big((\infty,0)\big)$ starts with~\eqref{e:psisigma} for some $k \ge 1$.
Then we have
\begin{align*}
\psi_\beta^{|\sigma_2^{k+1}(3)|}\big((\infty,0)\big) & = \psi_\beta^{|\sigma_2^k(3)|+3}\big(\big(\big\lfloor|\sigma_2^k(3)|/2\big\rfloor, \big\lceil|\sigma_2^{k-1}(3)|/2\big\rceil\big)\big) \cdots \\
& = \cdots \psi_\beta^{|\sigma_2^k(3)|+2}\big(\big(0, \big\lceil|\sigma_2^k(3)|/2\big\rceil\big)\big) \\
& = \psi_\beta^{|\sigma_2^k(3)|+1}\big(\big(\big\lceil|\sigma_2^k(3)|/2\big\rceil, \infty\big)\big) \cdots \\
& = \cdots \psi_\beta^{|\sigma_2^k(3)|}\big(\big(\infty,0\big)\, \big(0,\big\lceil|\sigma_2^k(3)|/2\big\rceil+1\big)\big) \\
& = \big(2\big\lceil|\sigma_2^k(3)|/2\big\rceil, \big\lceil|\sigma_2^k(3)|/2\big\rceil\big) \\ 
& \qquad \big(\big\lfloor|\sigma_2^k(3)|/2\big\rfloor, \big\lceil|\sigma_2^{k-1}(3)|/2\big\rceil\big) \cdots \big(\big\lfloor|\sigma_2^1(3)|/2\big\rfloor, \big\lceil|\sigma_2^0(3)|/2\big\rceil\big)\, \big(0,\infty\big) \cdots\,,
\end{align*}
where we have used the relations $a_{|\sigma_2^k(3)|} = 2 > 1 = a_{|\sigma_2^{k-1}(3)|+1}$, $a_1 > a_{|\sigma_2^k(3)|+1}$, and $a_1 \cdots a_{|\sigma_2^k(3)|} = a_{|\sigma_2^k(3)|+3} \cdots a_{2|\sigma_2^k(3)|+2}$.
Since $2\big\lceil|\sigma_2^k(3)|/2\big\rceil = |\sigma_2^k(3)|+1 = \big\lfloor|\sigma_2^{k+1}(3)|/2\big\rfloor$, we obtain inductively that $\psi_\beta^{|\sigma_2^k(3)|}\big((\infty,0)\big)$ starts with~\eqref{e:psisigma} for all $k \ge 0$.

For any $k \ge 0$, we have $a_{|\sigma_2^k(3)|} \ge 2$, $a_{|\sigma_2^k(3)|+1} = 1$, $a_{|\sigma_2^k(3)|+2} = 2$, and $t_k < 0$.
For any $k \ge 1$, this yields that
\[
L_\beta\big(\big(\big\lfloor|\sigma_2^k(3)|/2\big\rfloor, \big\lceil|\sigma_2^{k-1}(3)|/2\big\rceil\big)\big) = t_{|\sigma_2^{k-1}(3)|} - t_{|\sigma_2^k(3)|-1} > - \frac{1}{\beta} + \frac{2}{\beta^2} + \frac{t_0}{\beta^2} + \frac{2}{\beta} - \frac{1}{\beta^2} > \frac{1}{\beta}\,.
\]
Let $k' = -\big|\psi_\beta^{|\sigma_2^k(3)|}\big|$.
Then we have $u_{k'+j} = \big(\big\lfloor|\sigma_2^{k-j}(3)|/2\big\rfloor, \big\lceil|\sigma_2^{k-j-1}(3)|/2\big\rceil\big)$ for $0 \le j < k$, thus $(y_{k'}, y_{k'+k}) \cap \mathbb{Z}_{-\beta} = \emptyset$ and $y_{k'+k} - y_{k'} > k/\beta$.
Since $k \ge 1$ was chosen arbitrary, the distances between consecutive $(-\beta)$-integers are unbounded.
\end{proof}

Many other examples of sets $\mathbb{Z}_{-\beta}$ which are not relatively dense can be found by setting $a_1 a_2 \cdots = \sigma_3\, \sigma_2^\infty(3)$ with a morphism $\sigma_3$ such that $\sigma_3(2) = 2$, $\sigma_3(1) = 1$, and $\sigma_3(3)$ is a suitable word of odd length.

We conclude the paper by stating the following three open problems, for which partial solutions are given in this section.
Note that all the corresponding problems for positive bases have well-known, simple solutions, as mentioned above.
\begin{enumerate}
\item
Characterise the sequences $a_1 a_2 \cdots$ which are possible $(-\beta)$-expansions of~$\frac{-\beta}{\beta+1}$.
\item
Characterise the numbers $\beta > 1$ such that $\mathbb{Z}_{-\beta}$ is uniformly discrete.
\item
Characterise the numbers $\beta > 1$ such that $\mathbb{Z}_{-\beta}$ is relatively dense.
\end{enumerate}

\bibliographystyle{eptcs}
\bibliography{delone}
\end{document}